\documentclass{article}

\usepackage[english]{babel}
\usepackage{amsmath,amssymb,amsthm}
\usepackage[pdfpagemode=UseNone,bookmarks=false,pdfpagelayout=SinglePage,pdfstartview={XYZ null null 1},
pdfhighlight=/P,colorlinks=true,linkcolor=blue,citecolor=blue,urlcolor=blue]{hyperref}

\advance\textwidth 20mm  \advance\oddsidemargin -10mm
\advance\textheight 20mm \advance\topmargin -15mm

\theoremstyle{plain}
\newtheorem{theorem}{Theorem}
\newtheorem{proposition}[theorem]{Proposition}
\theoremstyle{definition}
\newtheorem{remark}{Remark}

\title{Non-Abelian Toda lattice and analogs of Painlev\'e III equation}

\date{March 17, 2022}

\author{V.E.\:Adler\thanks{L.D.\:Landau Institute for Theoretical Physics, Akademika Semenova 1A, 142432 Chernogolovka, Russian Federation (permanent address). E-mail: adler@itp.ac.ru} $^\ddag$,
M.P.\:Kolesnikov\thanks{Moscow Institute of Physics and Technology, Institutskiy 9, 141701 Dolgoprudny, Russian Federation (permanent address).} \thanks{Institute of Mathematics, Ufa Federal Research Centre, Russian Academy of Sciences, Chernyshevsky 112, 450008 Ufa, Russian Federation}}

\begin{document}
\maketitle

\begin{abstract}
In integrable models, stationary equations for higher symmetries serve as one of the main sources of reductions consistent with dynamics. We apply this method to the non-Abelian two-dimensional Toda lattice. It is shown that already the stationary equation of the simplest higher flow gives a non-trivial non-autonomous constraint that reduces the Toda lattice to a non-Abelian analog of the pumped Maxwell--Bloch equations. The Toda lattice itself is interpreted as an auto-B\"acklund transformation acting on the solutions of this system. Further self-similar reduction leads to non-Abelian analogs of the Painlev\'e III equation.
\medskip

\end{abstract}

\section{Introduction}

The non-Abelian Toda lattice \cite{Mikhailov_1981}
\[
 (g_{n,x}g^{-1}_n)_y=g_{n+1}g^{-1}_n-g_ng^{-1}_{n-1}
\]
is one of the fundamental three-dimensional models (two continuous independent variables $x,y$ and one discrete variable $n$). If the field variables are scalar then the substitution $g_n=e^{u_n}$ is possible which leads to the equation $u_{n,xy}=e^{u_{n+1}-u_n}-e^{u_n-u_{n-1}}$ \cite{Mikhailov_1979}. In the non-Abelian setting, we assume that $g_n$ are elements of an arbitrary non-commutative algebra ${\cal A}$ with identity element 1 and the operation of taking the inverse (for example, a matrix algebra). The polynomial form of the Toda lattice \cite{Salle_1982} 
\begin{equation}\label{TL}
 f_{n,y}=p_n-p_{n+1},\quad p_{n,x}=f_np_n-p_nf_{n-1}
\end{equation}
is obtained by the substitution $p_n=-g_ng^{-1}_{n-1}$ and $f_n=g_{n,x}g^{-1}_n$. The aim of our work is to construct a reduction of this lattice equation to non-Abelian analogs of the Painlev\'e equation P$_3$
\begin{equation}\label{P3}
 w''=\frac{(w')^2}{w}-\frac{w'}{z} +\frac{1}{z}(\alpha w^2+\beta) +\gamma w^3+\frac{\delta}{w}.
\end{equation}
The reduction is carried out in two stages. First, we reduce (\ref{TL}) to a two-dimensional system. The easiest way to do this is related with the 2-periodicity condition, which leads to the non-Abelian sinh--Gordon equation \cite{Salle_1982,Li_Nimmo_2008}. In the scalar case, other boundary conditions are also known, leading to exponential systems associated with Cartan matrices for simple Lie algebras \cite{Leznov_Saveliev_1980}. The reduction we use is of a different type and, as far as we know, has not been considered before (although it leads to some well-known equations). It has the form of some non-autonomous constraint and is related to the stationary equation of the simplest higher symmetry from the Toda lattice hierarchy. This reduction is defined in Section \ref{s:reduction}. 
It leads to a two-dimensional system, which can be considered as a non-Abelian analog of the pumped Maxwell--Bloch equations \cite{Burtsev_Zakharov_Mikhailov_1987}. In this case, the auxiliary linear equations for the Toda lattice are transformed into the zero curvature representation with variable spectral parameter for this system, and the shift in $n$ is interpreted as a B\"acklund transformation, see Section \ref{s:MB}.

At the next stage (sections \ref{s:P30} and \ref{s:P3}), we apply self-similar reductions, which lead to third order ODE systems of Painlev\'e type. In the non-Abelian case, these systems have no global first integrals; however, the reduction of the order of equations is possible on some special invariant submanifolds. This brings to non-Abelian analogs of P$_3$ with full and degenerate sets of parameters. The corresponding isomonodromic Lax pairs and B\"acklund transformations are also derived. These results generalize the scalar reductions described in \cite{Winternitz_1992,Burtsev_1993,Schief_1994,Clarkson_Mansfield_Milne_1996}.

It should be noted that analogs of the Painlev\'e equations under study are written not as a single second order equation, but as a coupled system of two first order equations. If the variables are scalar, then the exclusion of one of them and some additional transformations lead to P$_3$ in the standard form, but in the non-Abelian case this is not always possible (at least, if we restrict ourselves to the operations of addition, multiplication, and taking the inverse element). Non-Abelian analogs of the Painlev\'e equations are being actively studied, see for example \cite{Retakh_Rubtsov_2010,Kawakami_2015}. It is known that a scalar equation can admit several non-equivalent non-Abelian analogs. Some classification results based on the Kovalevskaya--Painlev\'e test and other approaches can be found in recent papers \cite{Adler_Sokolov_2021, Bobrova_Sokolov_2021a, Bobrova_Sokolov_2021b}. Here we do not try to reproduce all known analogs of P$_3$ and restrict ourselves to only those that are obtained as a result of our reduction.

\section{Reduction of the Toda lattice}\label{s:reduction}

We start by defining a constraint that is compatible with the Toda lattice. The following assertion is central to our construction, so we present it with a detailed proof. We then explain the origin of this reduction and show how it extends to linear equations for $\psi$-functions.

\begin{theorem}\label{th:ffp}
The non-Abelian Toda lattice (\ref{TL}) admits the constraint
\begin{equation}\label{ffp}
 f_{n-1}+f_n=\mu_np^{-1}_n,\quad \mu_n:=\varepsilon n+\mu_0,\quad \varepsilon,\mu_0\in\mathbb C.
\end{equation}
\end{theorem}

\begin{proof} 
1) The case $\varepsilon=\mu_0=0$ amounts to the 2-periodicity condition (since the differentiation of the equality $f_{n-1}+f_n=0$ with respect to $y$ implies $p_{n-1}=p_{n+1}$). Suppose now that $\mu_n\not\equiv0$, then the equality $\mu_n=0$ is possible for at most one value of $n$. 

2) Consider first the case $\mu_n\ne0$ for all $n$. Then the constraint (\ref{ffp}) leads to a separation of variables and the Toda lattice reduces to the pair of two-dimensional equations
\begin{gather}
\label{fx}
  f_{n,x}+f_{n+1,x}=f^2_n-f^2_{n+1},\\
\label{fy}
 f_{n,y}=\mu_n(f_{n-1}+f_n)^{-1}-\mu_{n+1}(f_n+f_{n+1})^{-1}.
\end{gather}
Indeed,
\begin{equation}\label{ffx}
\begin{aligned}
 f_{n-1,x}+f_{n,x}&=\mu_n(p^{-1}_n)_x=-\mu_np^{-1}_nf_n+\mu_nf_{n-1}p^{-1}_n\\
  &=-(f_{n-1}+f_n)f_n+f_{n-1}(f_{n-1}+f_n)=f^2_{n-1}-f^2_n.
\end{aligned}
\end{equation} 
Conversely, if the variables $f_n$ satisfy the pair (\ref{fx}) and (\ref{fy}) then we obtain a solution of the Toda lattice by setting $p_n=\mu_n(f_{n-1}-f_n)^{-1}$. To prove the theorem, it suffices to verify that equations (\ref{fx}) and (\ref{fy}) are consistent, which is expressed by the equality
\[
 D_y(f^2_n-f^2_{n+1})=D_x\bigl(\mu_n(f_{n-1}+f_n)^{-1}-\mu_{n+2}(f_{n+1}+f_{n+2})^{-1}\bigr).
\]
We write this as $a_n-a_{n+1}=0$, where
\[
 a_n= D_y(f^2_n)-D_x\bigl(\mu_n(f_{n-1}+f_n)^{-1}+\mu_{n+1}(f_n+f_{n+1})^{-1}\bigr).
\]
First, we transform one term from the right side. From (\ref{ffx}) it follows
\begin{align*}
 -D_x\bigl((f_{n-1}+f_n)^{-1}\bigr)= f_{n-1}(f_{n-1}+f_n)^{-1}-(f_{n-1}+f_n)^{-1}f_n.
\end{align*}
Then
\begin{align*}
 a_n&= f_n\bigl(\mu_n(f_{n-1}+f_n)^{-1}-\mu_{n+1}(f_n+f_{n+1})^{-1}\bigr)\\ 
    &\quad +\bigl(\mu_n(f_{n-1}+f_n)^{-1}-\mu_{n+1}(f_n+f_{n+1})^{-1}\bigr)f_n\\ 
    &\quad +\mu_n\bigl(f_{n-1}(f_{n-1}+f_n)^{-1} - (f_{n-1}+f_n)^{-1}f_n\bigr)\\
    &\quad +\mu_{n+1}\bigl(f_n(f_n+f_{n+1})^{-1} - (f_n+f_{n+1})^{-1}f_{n+1}\bigr)\\
 &= \mu_n-\mu_{n+1}
\end{align*}
(here and below we identify the number $\mu\in{\mathbb C}$ with the element $\mu 1\in\cal A$, where $1$ is the unit element of the noncommutative algebra), therefore $a_n-a_{n+1}=\mu_n-2\mu_{n+1}+\mu_{n+2}$, which is 0 for $\mu_n=\varepsilon n+\mu_0$.

3) If $\mu_k=0$ for some $k$ then $\varepsilon\ne0$ and $\mu_0=-\varepsilon k$. We can set $\mu_n=n$ without loss of generality, by scaling and changing $n\to n-k$. Then it is easy to see that the solution satisfies the reflection symmetry
\begin{equation}\label{refl}
 f_{-n}=-f_{n-1},\quad p_{-n}=p_n,\quad n=1,2,\dots
\end{equation}
with the lattice equations (\ref{fy}) and (\ref{fx}) restricted to the half-line $n=1,2,\dots$:
\begin{equation}\label{half-line}
\begin{gathered} 
 f_{0,y}=p_0-(f_0+f_1)^{-1},\quad f_{n,y}=n(f_{n-1}+f_n)^{-1}-(n+1)(f_n+f_{n+1})^{-1},\\ 
 p_{0,x}=f_0p_0+p_0f_0,\quad f_{n-1,x}+f_{n,x}=f^2_{n-1}-f^2_n.
\end{gathered}
\end{equation}
The compatibility of these equations is proved in the same way as before, except that when checking the equation $(f_0+f_1)_{xy}=(f_0+f_1)_{yx}$, the equality arises
\[
 \bigl(p_0-2(f_1+f_2)^{-1}\bigr)_x=(f^2_0-f^2_1)_y.
\]
It is easy to verify that it holds identically.
\end{proof}

\begin{remark}
It is possible to write the constraint (\ref{ffp}) in the form $p_n=\mu_n(f_{n-1}-f_n)^{-1}$, but this is not quite equivalent, since in the case 3) for $\mu_n=n$ it leads to the boundary condition $p_0=0$ instead of $f_{-1}+f_0=0$. Then the equation $p_{0,x}=f_0p_0-p_0f_{-1}$ holds automatically, while the reflection conditions (\ref{refl}) are not necessary. As a result, equations (\ref{fy}) and (\ref{fx}) turn into two independent sybsystems, one for the variables $f_0,f_1,\dots$ and another for $f_{-1},f_{-2},\dots$. Each of these subsystems is equivalent to (\ref{half-line}) with $p_0=0$. Therefore, the constraint in the form (\ref{ffp}) is a bit more general. The equation $(f_{-1}+f_0)p_0=0$ may also have solutions with zero divisors, but we will not analyze the resulting boundary conditions.
\end{remark}

Note that (\ref{fx}) and (\ref{fy}) are themselves well-known integrable lattice equations. Equation (\ref{fx}) defines B\"acklund transformations with zero parameters for the equation $f_t=f_{xxx}-3f^2f_x-3f_xf^2$ which is one of two non-Abelian versions of the modified KdV equation. The substitution $v_n=p_n/\mu_n$ turns (\ref{fy}) into
\[
 v_{n,y}=v_n(\mu_{n+1}v_{n+1}-\mu_{n-1}v_{n-1})v_n.
\]
For $\mu_n=1$, this is the modified Volterra lattice (again, one of two non-Abelian versions), and for $\mu_n=n$ this is its master-symmetry. Master-symmetries for scalar equations of Volterra lattice type were studied in \cite{Oevel_Zhang_Fuchssteiner_1989, Cherdantsev_Yamilov_1995}, some non-Abelian generalizations and reductions to the Painlev\'e equations were found in \cite{Adler_2020}. In the scalar case, the consistency of equations (\ref{fx}) and (\ref{fy}) was noted in \cite{Shabat_Yamilov_1991} (for $\mu_n=1$) and in \cite{Garifullin_Habibullin_Yamilov_2015} (for $\mu_n=n$), but the fact that they are both embedded in the two-dimensional Toda lattice went unnoticed. Similar results are also known regarding the compatibility of differential-difference equations with discrete equations on a square lattice, see e.g.~\cite{Xenitidis_2011}.

Let us demonstrate that the reduction (\ref{ffp}) is related to the stationary equation for the simplest higher symmetry of the Toda lattice. Recall that this lattice itself serves as a compatibility condition for the linear equations
\begin{equation}\label{DxDy}
 \psi_{n,x}= \psi_{n+1}+f_n\psi_n,\quad \psi_{n,y}= p_n\psi_{n-1},
\end{equation}
and its symmetries are defined as compatibility conditions of (\ref{DxDy}) with equations of the form
\[
 \psi_{n,t_k}=\psi_{n+k}+h^{(k,1)}_n\psi_{n+k-1}+\dots+h^{(k,k)}_n\psi_n.
\] 
The gauge $h^{(k,1)}_n=f_n+\dots+f_{n+k-1}$ can be chosen without loss of generality. In particular, for $k=2$ we have the equation
\begin{equation}\label{Dt}
 \psi_{n,t}=\psi_{n+2}+(f_n+f_{n+1})\psi_{n+1}+h_n\psi_n
\end{equation}
and the conditions of its compatibility with (\ref{DxDy}) are of the form
\begin{gather}
\label{xt}
 f_{n,x}+f_{n+1,x}=f^2_n-f^2_{n+1}-h_n+h_{n+1},\quad h_{n,x}=f_{n,t}+[f_n,h_n],\\
\label{yt}
 h_{n,y}=p_n(f_{n-1}+f_n)-(f_n+f_{n+1})p_{n+1},\quad p_{n,t}=h_np_n-p_nh_{n-1}.
\end{gather}
It is easy to see that if all field variables are commutative, then the requirement that $f_n$ and $p_n$ be independent of $t$ implies that $h_n$ is independent of $x$ and $n$, and then equations are reduced to (\ref{fx}) and (\ref{ffp}), up to a transformation of $y$. Thus, in the scalar case the constraint (\ref{ffp}) is equivalent to the stationary equation for the higher symmetry. In the non-Abelian setting, the stationary equation is more general, but it still holds true by virtue of the constraint (\ref{ffp}), if we choose $h_n=-\varepsilon y-\kappa\in{\mathbb C}$. 

Equations (\ref{DxDy}) and the stationary equation (\ref{Dt}) 
\[
 \psi_{n+2}+(f_{n+1}+f_n)\psi_{n+1}=(\varepsilon y+\kappa)\psi_n
\]
can be transformed into zero curvature representations for the reduction under study, with $\kappa$ playing the role of a spectral parameter. Indeed, let us write these equations in the matrix form
\[
 \Psi_{n,x}=U_n\Psi_n,\quad \Psi_{n,y}=V_n\Psi_n,\quad \Psi_{n+1}=W_n\Psi_n, 
\]
where $\Psi_n=(\psi_n,\psi_{n+1})^t$ and
\[
 U_n=\begin{pmatrix}
  f_n & 1\\
  \varepsilon y+\kappa & -f_n
 \end{pmatrix},\quad
 V_n=\begin{pmatrix}
  \dfrac{\mu_n}{\varepsilon y+\kappa} & \dfrac{p_n}{\varepsilon y+\kappa}\\
  p_{n+1} & 0
 \end{pmatrix},\quad
 W_n=\begin{pmatrix}
  0 & 1\\
 \varepsilon y+\kappa & -f_n-f_{n+1}
 \end{pmatrix}.
\]
Then the compatibility conditions are
\[
 W_{n,x}=U_{n+1}W_n-W_nU_n,\quad W_{n,y}=V_{n+1}W_n-W_nV_n,\quad U_{n,y}=V_{n,x}+[V_n,U_n]
\]
and it is easy to check that these equations are equivalent to (\ref{TL}) and (\ref{ffp}). The above matrices can be brought to a more symmetric form by introducing the variable spectral parameter $\lambda^2=\varepsilon y+\kappa$ (this leads to appearance of a term with $\partial_\lambda$ in the linear equations for $\psi$-functions) \cite{Burtsev_Zakharov_Mikhailov_1987}. Simple calculations bring to the following representation.

\begin{proposition}\label{pr:zcr}
The Toda lattice equations (\ref{TL}) with the constraint (\ref{ffp}) are equivalent to equations
\begin{equation}\label{zcr}
 \begin{gathered}
 W_{n,x}=U_{n+1}W_n-W_nU_n,\quad 
 W_{n,y}+\frac{\varepsilon}{2\lambda}W_{n,\lambda}=V_{n+1}W_n-W_nV_n,\\
 U_{n,y}+\frac{\varepsilon}{2\lambda}U_{n,\lambda}=V_{n,x}+[V_n,U_n]
 \end{gathered}
\end{equation}
with the matrices
\begin{equation}\label{UVW}
 U_n=\begin{pmatrix}
  f_n & \lambda\\
  \lambda & -f_n
 \end{pmatrix},\quad
 V_n=\lambda^{-1}\begin{pmatrix}
  \frac{\mu_n+\mu_{n+1}}{2\lambda} & p_n\\
  p_{n+1} & 0
 \end{pmatrix},\quad
 W_n=\begin{pmatrix}
  0 & \lambda\\
 \lambda & -f_n-f_{n+1}
 \end{pmatrix}.
\end{equation}
\end{proposition}

\section{Partial differential systems}\label{s:MB}

Due to the constraint (\ref{ffp}), the Toda lattice (\ref{TL}) turns into a closed system for the variables $f=f_n$, $p=p_n$, $q=p_{n+1}$ and parameters $\mu=\mu_n$, $\nu=\mu_{n+1}$, for any $n$: 
\begin{equation}\label{fpq}
 f_y=p-q,\quad p_x=fp+pf-\mu,\quad q_x=-fq-qf+\nu
\end{equation}
(recall that $\mu$ and $\nu$ are understood as scalars multiplied by $1\in{\cal A}$). Denoting the shift action $n\mapsto n+1$ in the lattice with a tilde, we reformulate Theorem \ref{th:ffp} and Proposition \ref{pr:zcr} as the following statements, which can also be verified by direct calculations.

\begin{proposition}
The B\"acklund transformation
\begin{equation}\label{fpq.BT}
 \tilde p=q,\quad \tilde q=p+\nu q^{-1}q_yq^{-1},\quad \tilde f=-f+\nu q^{-1},\quad 
 \tilde\mu=\nu,\quad \tilde\nu=-\mu+2\nu
\end{equation} 
maps the solution $f,p,q$ of the system (\ref{fpq}) with parameters $\mu,\nu$ into the solution $\tilde f,\tilde p,\tilde q$ of the same system with parameters $\tilde\mu,\tilde\nu$.
\end{proposition}

\begin{proposition}\label{pr:fpq.zcr}
The system (\ref{fpq}) admits the representation
\begin{equation}\label{fpq.UV}
 U_y+\frac{\nu-\mu}{2\lambda}U_\lambda=V_x+[V,U],
\end{equation}
and its B\"acklund transformation (\ref{fpq.BT}) admits the representation
\begin{equation}\label{fpq.W}
 W_x=\widetilde UW-WU,\quad W_y+\frac{\nu-\mu}{2\lambda}W_\lambda=\widetilde VW-WV,
\end{equation}
where
\begin{equation}\label{fpq.UVW}
 U=\begin{pmatrix}
  f & \lambda\\
  \lambda & -f
 \end{pmatrix},\quad
 V=\frac{1}{\lambda}\begin{pmatrix}
   \frac{\mu+\nu}{2\lambda} & p\\
   q & 0
 \end{pmatrix},\quad
 W=\begin{pmatrix}
   0 & \lambda\\
   \lambda & -\nu q^{-1}
 \end{pmatrix}.
\end{equation}
\end{proposition}

In the commutative case, the elimination of $p$ brings the system (\ref{fpq}) to the equation
\[
 ff_{xxy}=f_xf_{xy}+4f^3f_y+(\mu+\nu)f_x+2(\nu-\mu)f^2.
\]
For $\nu=\mu$, this is the so-called ``negative'' symmetry of the mKdV equation $f_t=f_{xxx}-6f^2f_x$. On the other hand, we can compare (\ref{fpq}) with the pumped Maxwell--Bloch system introduced in \cite{Burtsev_Zakharov_Mikhailov_1987}:
\[
 E_y=\rho,\quad \rho_x=NE,\quad 2N_x=-\rho^*E-\rho E^*+2c,\quad E,\rho\in{\mathbb C},~~ N\in{\mathbb R},
\]
where $c$ is the pumping parameter. For $\rho,E\in{\mathbb R}$, these equations are simplified to the system
\begin{equation}\label{PMB-real}
 E_y=\rho,\quad \rho_x=NE,\quad N_x=-\rho E+c
\end{equation}
studied in \cite{Winternitz_1992,Burtsev_1993,Schief_1994,Clarkson_Mansfield_Milne_1996} where it was shown that it admits a selfsimilar reduction to the P$_3$ equation. The system (\ref{PMB-real}) and the scalar system (\ref{fpq}) are related by the simple change
\[
 2f=iE,\quad 4p=N+i\rho,\quad 4q=N-i\rho,\quad c=4\mu=-4\nu.
\]
Thus, the system (\ref{fpq}) can be viewed as a non-Abelian generalization both for the negative mKdV flow with an additional parameter $\nu-\mu$ and for the real pumped Maxwell--Bloch system with an additional parameter $\nu+\mu$. 

The system (\ref{fpq}) with $\mu=\nu=0$
\begin{equation}\label{fpq0}
 f_y=p-q,\quad p_x=fp+pf,\quad q_x=-fq-qf
\end{equation}
is a degenerate case corresponding to the constraint (\ref{ffp}) with $\mu_n=0$, that is, to the 2-periodic boundary condition
\[
 f_{2n}=f,\quad f_{2n+1}=-f,\quad p_{2n}=p,\quad p_{2n+1}=q.
\]

\begin{remark}
The general 2-periodicity condition $f_{n+2}=f_n$, $p_{n+2}=p_n$ is also equivalent to (\ref{fpq0}). Indeed, we have $(f_n+f_{n+1})_y=p_n-p_{n+2}=0$, that is $f_n+f_{n+1}=a(x)$, and it is possible to set $a=0$ without loss of generality by the gauge transformation $f_n=A\tilde f_nA^{-1}+A_xA^{-1}$, $p_n=A\tilde p_nA^{-1}$ with $2A_xA^{-1}=a$.
\end{remark}

For the system (\ref{fpq0}), an additional constraint $pq=\beta(y)\in{\mathbb C}$ is possible due to the relation $(pq)_x=[f,pq]$; moreover, one can set $\beta=1$ by the change $(p,q,\partial_y)\to \beta^{1/2}(p,q,\partial_y)$. This brings to a more special reduction
\[
 f_{2n}=f,\quad f_{2n+1}=-f,\quad p_{2n}=p,\quad p_{2n+1}=p^{-1}
\]
and the non-Abelian sinh-Gordon equation \cite{Salle_1982} 
\begin{equation}\label{naSG}
 f_y=p-p^{-1},\quad p_x=fp+pf.
\end{equation}
The mapping (\ref{fpq.BT}) in this case turns into the trivial change $\tilde f=-f$, $\tilde p=p^{-1}$. The zero curvature representation for (\ref{naSG}) is of the form
\begin{equation}\label{fp.zcr}
 U_y-V_x=[V,U],\quad
 U=\begin{pmatrix}
  f & \lambda\\
  \lambda & -f
 \end{pmatrix},\quad
 V=\lambda^{-1}\begin{pmatrix}
   0 & p\\
  p^{-1} & 0
 \end{pmatrix}.
\end{equation}
Notice, that the systems (\ref{fpq0}) and (\ref{naSG}) with scalar variables are equivalent since the relation  $pq=\beta(y)$ is the first integral. In the non-Abelian case, this is only a partial first integral, so the system (\ref{fpq0}) is more general than (\ref{naSG}). Another difference is that the scalar system (\ref{naSG}) easily reduces to the sinh-Gordon equation in rational form
\begin{equation}\label{SGp}
 p_{xy}=\frac{p_xp_y}{p}+2p^2-2,
\end{equation} 
while in the non-Abelian case the elimination of $f$ only by use of non-commutative algebra operations is impossible.

\section{Sinh-Gordon equation and non-Abelian analog of P$^{(8)}_3$}\label{s:P30}

A self-similar reduction of the general system (\ref{fpq}) brings to a system of three first-order ODEs. In the following sections, we show that this system has a partial first integral (which did not exist before the reduction), which allows us to reduce the order and to obtain an analog of the P$_3$ equation with a full set of parameters. In this section, we start with a simpler case of the sinh-Gordon equation (\ref{naSG}), which is already second order.

Obviously, the scalar equation (\ref{SGp}) is invariant under the Lorentz group $(x,y)\mapsto(\varepsilon x,y/\varepsilon)$. Therefore, the self-similar substitution is possible
\[
 p(x,y)=p(z),\quad z=-2xy,
\]
which brings to equation (\ref{P3}) with parameters $\alpha=1$, $\beta=-1$, $\gamma=\delta=0$ for the variable $p(z)$:
\begin{equation}\label{P31100}
 p''= \frac{(p')^2}{p}-\frac{p'}{z}+\frac{p^2-1}{z}.
\end{equation}
This is the simplest and most well-studied case of P$_3$ known as P$^{(8)}_3$ equation (the classification of different cases of P$_3$ is given, for example, in \cite{Clarkson_2019}, see also \cite{Its_Novokshenov_1986, GLS, Fokas_Its_Kapaev_Novokshenov_2006}). 

In the non-Abelian case, the reduction is essentially just as simple, but it can be generalized slightly by adding the conjugation by elements of the form
\[
 y^a:=\exp(a\log y),\quad a\in\cal A,
\]
where $a$ is an arbitrary non-Abelian constant. An analog of equation (\ref{P31100}) is the system (\ref{fp0.z}) from the following Proposition.

\begin{proposition}
The non-Abelian sinh-Gordon equation (\ref{naSG}) admits the self-similar reduction
\[
 p(x,y)= y^ap(z)y^{-a},\quad f(x,y)= -2y^{1+a}f(z)y^{-a},\quad z=-2xy,\quad a\in\cal A,
\]
where $f(z)$ and $p(z)$ satisfy the equations
\begin{equation}\label{fp0.z}
 zf'= \frac{1}{2}(p-p^{-1})-f-[a,f],\quad p'=fp+pf.
\end{equation}
This system admits the isomonodromic Lax pair $A'=B_\zeta+[B,A]$ with the matrices
\[
 A=\begin{pmatrix}
  (a+zf)/\zeta & p/\zeta^2-z/2 \\
  p^{-1}/\zeta^2-z/2 & (a-zf)/\zeta 
 \end{pmatrix},\quad 
 B=\begin{pmatrix}
  f & -\zeta/2 \\
  -\zeta/2 & -f 
 \end{pmatrix}. 
\]
\end{proposition}
\begin{proof}
Equations (\ref{fp0.z}) are obtained straightforwardly. In order to obtain the matrices $A$ and $B$, we apply an additional change of the spectral parameter $\lambda=y\zeta$, then the matrices (\ref{fp.zcr}) take the form
\[
 U= y^{1+a}\widetilde Uy^{-a},\quad
 \widetilde U = \begin{pmatrix}
   -2f & \zeta \\
   \zeta & 2f 
  \end{pmatrix},\quad
 V= -\zeta^{-1}y^{-1+a}\widetilde Vy^{-a},\quad
 \widetilde V = \begin{pmatrix}
   0 & p \\
   p^{-1} & 0 
  \end{pmatrix},
\] 
and the derivatives are replaced according to the rule
\[
 \partial_x \to -2y\partial_z,\quad
 \partial_y \to \partial_y+\frac{z}{y}\partial_z-\frac{\zeta}{y}\partial_\zeta.
\]
Then the dependence on $y$ in the equation $U_y=V_x+[V,U]$ is canceled out and it takes the form
\[
 (z\widetilde U-2\zeta^{-1}\widetilde V)' = \zeta\widetilde U_\zeta + [\widetilde U,\zeta^{-1}\widetilde V+a].
\] 
The above Lax pair appears as a result of the changes $-2B=\widetilde U$ and $A=\zeta^{-2}\widetilde V+\zeta^{-1}(zB+a)$.
\end{proof}

\section{Self-similar reduction of system (\ref{fpq})}\label{s:P3}

\subsection{Third order ODE system}

The scaling group for the system (\ref{fpq}) is different from the group for the sinh-Gordon equation: for (\ref{naSG}) the homogeneity weights are $\rho(\partial_x)=-\rho(\partial_y)=\rho(f)=1$ and $\rho(p)=0$, while for (\ref{fpq}), as it is easy to see,
\[
 \rho(\partial_x)=1,\quad \rho(\partial_y)=-2,\quad \rho(f)=1,\quad \rho(p)=\rho(q)=-1.
\]
Accordingly, an independent self-similar variable should be $z=xy^{1/2}$ rather than $xy$ and we arrive at the self-similar substitution
\begin{equation}\label{fpq.sub}
\begin{gathered}
 f(x,y)= y^{1/2-a}f(z)y^a,\quad p(x,y)= y^{-1/2-a}p(z)y^a,\\
 q(x,y)= y^{-1/2-a}q(z)y^a,\quad z=xy^{1/2},\quad a\in{\cal A}.
\end{gathered}  
\end{equation}

\begin{proposition}
The reduction (\ref{fpq.sub}) in equations (\ref{fpq}) brings to the system
\begin{equation}\label{fpq.z}
 (zf)'= 2p-2q+2[a,f],\quad p'= fp+pf-\mu,\quad q'= -fq-qf+\nu,
\end{equation}
which is invariant under the B\"acklund transformation
\begin{equation}\label{fpq.zBT}
\begin{gathered}
 \tilde p=q,\quad 
 \tilde q=p-\frac{\nu z}{2}(fq^{-1}+q^{-1}f)-\frac{\nu}{2}q^{-1}+\frac{\nu^2z}{2}q^{-2}+\nu[a,q^{-1}],\\
 \tilde f=-f+\nu q^{-1},\quad 
 \tilde\mu=\nu,\quad \tilde\nu=-\mu+2\nu.
\end{gathered}
\end{equation}
\end{proposition}
\begin{proof}
Equations (\ref{fpq.z}) are obtained straightforwardly. The transformation (\ref{fpq.BT}) preserves the homogeneity with respect to the above weights. Therefore, if a solution $f(x,y)$, $p(x,y)$, $q(x,y)$ of the system (\ref{fpq}) possesses the self-similar structure (\ref{fpq.sub}), then this is true also for the new solution $\tilde f(x,y)$, $\tilde p(x,y)$, $\tilde q(x,y)$. Hence, it is also described by a system of the form (\ref{fpq.z}) and we only have to rewrite equations (\ref{fpq.BT}) under the reduction (\ref{fpq.sub}).
\end{proof}

Applying the reduction to equations (\ref{fpq.UV}) and (\ref{fpq.W}), we obtain the following proposition. Note that the representations (\ref{fpq.AB}), (\ref{fpq.K}) can be brought to the standard form with unit coefficient at $\partial_\zeta$ by dividing $A$ by $\zeta^2-\nu+\mu$.

\begin{proposition}
The system (\ref{fpq.z}) admits the isomonodromic Lax pair
\begin{equation}\label{fpq.AB}
  A'=(\zeta^2-\nu+\mu)B_\zeta+[B,A]
\end{equation}
and the transformation (\ref{fpq.zBT}) is equivalent to the equations
\begin{equation}\label{fpq.K}
 K'=\widetilde BK-KB,\quad (\zeta^2-\nu+\mu)K_\zeta=\widetilde AK-KA
\end{equation}
with the matrices
\begin{equation}\label{fpq.ABK}
\begin{gathered}
 A=\begin{pmatrix}
  \zeta zf -2\zeta a-\frac{\mu+\nu}{\zeta}+\zeta\kappa & \zeta^2z-2p\\
  \zeta^2z-2q & -\zeta zf -2\zeta a+\zeta\kappa 
 \end{pmatrix},\\
 B=\begin{pmatrix}
  f & \zeta \\
  \zeta & -f 
 \end{pmatrix},\quad 
 K=\begin{pmatrix}
   0 & \zeta \\
   \zeta & -\nu q^{-1} 
  \end{pmatrix}, 
\end{gathered}
\end{equation}
where $\kappa$ in $A$ is an additional scalar parameter such that $\tilde\kappa=\kappa+1$.
\end{proposition}
\begin{proof}
We extend the substitution (\ref{fpq.sub}) with the relation $\lambda=y^{1/2}\zeta$, then $y$ in the matrices (\ref{fpq.UVW}) is separated out:
\[
 U= y^{1/2-a}By^a,\quad V= y^{-1-a}Cy^a,\quad W=y^{1/2-a}Ky^a,\quad
 C=\frac{1}{\zeta} 
  \begin{pmatrix}
   \frac{\mu+\nu}{2\zeta} & p \\
    q & 0 
  \end{pmatrix}.
\]
The derivatives are replaced according to the rule
\[
 \partial_x \to y^{-1/2}\partial_z,\quad
 \partial_y\to\partial_y +\frac{z}{2y}\partial_z -\frac{\zeta}{2y}\partial_\zeta,\quad
 \partial_\lambda\to y^{-1/2}\partial_\zeta.
\]
As a result, equation (\ref{fpq.UV}) transforms to
\[
 \zeta(zB-2C)'= (\zeta^2-\nu+\mu)B_\zeta -[B,2\zeta(C+a)].
\]
In order to bring this relation to the form (\ref{fpq.AB}) we only have to denote $A=\zeta(zB-2C-2a+\kappa)$. Here $\kappa$ is an arbitrary parameter, since it cancels out in (\ref{fpq.AB}). However, this additional term is neccessary in order to obtain the consistent representations (\ref{fpq.K}) for the B\"acklund transform. The equation for $K'$ is obtained from the first equation (\ref{fpq.W}) automatically, while the second equation yields
\[
 \frac{1}{2}(zK)'-[a,K] -\left(\frac{\zeta}{2}-\frac{\nu-\mu}{2\zeta}\right)K_\zeta =\widetilde CK-KC.
\]
After replacing $K'=\widetilde BK-KB$, this equation is reduced to the form
\[
 (\zeta^2-\nu+\mu)K_\zeta-\zeta K= (\widetilde A-\zeta\tilde\kappa)K-K(A-\zeta\kappa),
\]
which coincides with the second equation (\ref{fpq.K}) under the choice $\tilde\kappa=\kappa+1$.
\end{proof}

In the scalar case, the order of the system (\ref{fpq.z}) can be reduced by use of a first integral. We will demonstrate that in the non-Abelian case this is possible due to a partial first integral, that is, equations (\ref{fpq.z}) can be restricted to some invariant submanifold $J=0$ which is also preserved under the B\"acklund transformation (\ref{fpq.zBT}). This leads to a second order system, which is a non-Abelian analog of P$_3$. We study the cases of $\nu=\mu$ and $\nu\ne\mu$ separately (this corresponds to the constraint (\ref{ffp}) with $\varepsilon=0$ and $\varepsilon\ne0$).  

\subsection{The case $\nu=\mu$: a non-Abelian analog of P$^{(7)}_3$}\label{s:numu}

\begin{proposition}
The system (\ref{fpq.z}) with $\nu=\mu\ne0$ admits the invariant submanifold
\begin{equation}\label{fpq.J0}
 J(\kappa)= 2pq-\mu(zf-2a-\kappa)=0,
\end{equation}
which is mapped to the submanifold $\tilde J(\tilde\kappa)=0$ with $\tilde\kappa=\kappa+1$ under the transformation (\ref{fpq.zBT}).
\end{proposition}
\begin{proof}
The partial first integral (\ref{fpq.J0}) is easily derived: we have
\[
 (2pq)'=2[f,pq]+2\mu(p-q)=2[f,pq]+\mu(zf)'+2\mu[f,a] \quad\Rightarrow\quad J'=[f,J],
\]
which implies the invariance of the equation $J=0$. In a similar way, it is easy to check that $\tilde J(\tilde\kappa)=qJ(\kappa)q^{-1}$, which proves the invariance with respect to the B\"acklund transformation.
\end{proof}

On the level set $J=0$, we have $2q=\mu p^{-1}(zf-2a-\kappa)$ and (\ref{fpq.z}) turns into a second order system. A direct calculations bring to the following formulas.

\begin{proposition}
The system
\begin{equation}\label{fp1.z}
 (zf)'= 2p-\mu p^{-1}(zf-2a-\kappa)+2[a,f],\quad p'= fp+pf-\mu
\end{equation}
admits the B\"acklund transformation
\[
 \tilde p=\frac{\mu}{2}p^{-1}(zf-2a-\kappa),\quad 
 \tilde f=-f+\mu\tilde p^{-1},\quad
 \tilde\mu=\mu,\quad \tilde\kappa=\kappa+1
\]
and is equivalent to the Lax representation $A'=\zeta^2B_\zeta+[B,A]$, where
\[
 A=\begin{pmatrix}
  \zeta zf -2\zeta a-\frac{2\mu}{\zeta}+\zeta\kappa & \zeta^2z-2p\\
  \zeta^2z-\mu p^{-1}(zf-2a-\kappa) & -\zeta zf -2\zeta a +\zeta\kappa
 \end{pmatrix},\quad 
 B=
 \begin{pmatrix}
  f & \zeta \\
  \zeta & -f 
 \end{pmatrix}.
\]
\end{proposition}

It is easy to check that if $f,p$ and $a$ are scalars then the elimination of $f$ brings (\ref{fp1.z}) to the equation
\[
 p''=\frac{p'^2}{p}-\frac{p'}{z}+\frac{1}{z}(4p^2+\mu(4a-1+2\kappa))-\frac{\mu^2}{p},
\]
which is (\ref{P3}) with the values of parameters
\[
 \alpha=4,\quad \beta=\mu(4a-1+2\kappa),\quad \gamma=0,\quad \delta=-\mu^2,
\]
that is, the intermediate P$^{(7)}_3$ equation.

\subsection{The case $\nu\ne\mu$: a non-Abelian analog of P$^{(6)}_3$}

\begin{proposition}
The system (\ref{fpq.z}) with $\nu-\mu=\varepsilon\ne0$ admits the invariant submanifold
\begin{equation}\label{fpq.J1}
 J(\kappa)= 2q-\varepsilon z 
    +\varepsilon(zf+2a-\kappa)(2p-\varepsilon z)^{-1}(zf-2a+\kappa-2\mu/\varepsilon-1)=0,
\end{equation}
which is mapped to the submanifold $\tilde J(\tilde\kappa)=0$ with $\tilde\kappa=\kappa+1$ under the transformation (\ref{fpq.zBT}).
\end{proposition}
\begin{proof}
In this case the partial first integral is less obvious, but we can find it with the help of the representation (\ref{fpq.AB}). For $\zeta=\varepsilon^{1/2}$, the $2\times2$ matrix $A$ satisfies the Lax equation $A'=[B,A]$. It is easy to prove that the quasideterminant $|A|_{12}=a_{12}-a_{11}a^{-1}_{21}a_{22}$ is a partial first integral for such an equation and this gives the expression (\ref{fpq.J1}), up to a scalar factor. 

Moreover, for $\zeta=\varepsilon^{1/2}$ we have $\widetilde AK=KA$, where
$K=\left(\begin{smallmatrix} 0 & \zeta\\ \zeta & -\nu q^{-1}\end{smallmatrix}\right)$, and from here it is easy to obtain the relation $|\widetilde A|_{21}=|A|_{12}$. Taking into account the homological relations for quasideterminants, this proves that the submanifold $J=0$ is preserved under the B\"acklund transformation.
\end{proof}

\begin{remark}
For the case $\varepsilon=0$ from Section \ref{s:numu}, the above proof is not directly applicable, since $A$ contains a term with $\zeta^{-1}$. Nevertheless, the integral (\ref{fpq.J1}) can be defined also in this case by combining  $\varepsilon$ with the last factor and setting $\varepsilon=0$. This brings to the constraint
\[
 2q-(zf+2a-\kappa)p^{-1}\mu=0 \quad\Rightarrow\quad 2qp-\mu(zf+2a-\kappa)=0.
\]
It is slightly different from (\ref{fpq.J0}), but it also defines an invariant submanifold for (\ref{fpq.z}) with $\nu=\mu$. In order to obtain (\ref{fpq.J0}), we can apply the change $\kappa-2\mu/\varepsilon=-\hat\kappa$ (\ref{fpq.J1}). Since $\tilde\mu=\mu+\varepsilon$, the new parameter is also transformed according to the rule $\tilde{\hat\kappa}=\hat\kappa+1$, then (\ref{fpq.J1}) takes the form
\[
 2q-\varepsilon z +\varepsilon(zf+2a+\hat\kappa-2\mu/\varepsilon)(2p-\varepsilon z)^{-1}(zf-2a-\hat\kappa)=0,
\]
and passing to the limit $\varepsilon\to0$ in this expression, we obtain
\[
 2q -\mu p^{-1}(zf-2a-\hat\kappa)=0 \quad\Rightarrow\quad 2pq-\mu(zf-2a-\hat\kappa)=0.
\]
Thus, the constraint (\ref{fpq.J1}) contains two constraints of type (\ref{fpq.J0}) as the limiting cases. 
\end{remark}

The elimination of $q$ by use of (\ref{fpq.J1}) brings to the system (\ref{fp2.z}) below. As before, its Lax representation is obtained from the general formulas (\ref{fpq.AB}) and (\ref{fpq.K}) by replacing $q$; the matrix $A$ becomes rather unwieldy and we do not write it explicitly.

\begin{proposition}
The system
\begin{equation}\label{fp2.z}
\begin{aligned}
 (zf)'&= 2p-\varepsilon z 
     +\varepsilon(zf+2a-\kappa)(2p-\varepsilon z)^{-1}(zf-2a+\kappa-2\mu/\varepsilon-1)+2[a,f],\\
 p'&= fp+pf-\mu
\end{aligned} 
\end{equation}
is invariant with respect to the B\"acklund transformation
\begin{gather*}
 2\tilde p=\varepsilon z -\varepsilon(zf+2a-\kappa)(2p-\varepsilon z)^{-1}(zf-2a+\kappa-2\mu/\varepsilon-1),\\
 \tilde f=-f+(\mu+\varepsilon)\tilde p^{-1},\quad 
 \tilde\mu=\mu+\varepsilon,\quad \tilde\kappa=\kappa+1.
\end{gather*}
\end{proposition}

This system gives an analog of P$^{(6)}_3$ equation with generic parameters. In order to demonstrate this, let us compare it with (\ref{P3}) assuming that all variables are scalar. In this case we can set $a=0$ without loss of generality (this is equivalent to changing of $\kappa$), which gives the system
\begin{equation}\label{fp2.scalar}
 (zf)'= 2p-\varepsilon z +\frac{\varepsilon(zf-\kappa)(zf+\kappa-2\mu/\varepsilon-1)}{2p-\varepsilon z},\quad
 p'= 2fp-\mu.
\end{equation}
It is clear that the variable $f$ can be eliminated, but the calculations here turn out to be more complicated than in the previous sections, where P$_3$ was obtained directly for the variable $p$. Comparing the terms with derivatives with the canonical list of Painlev\'e equations, one can find a point change
\[
 p(z)=\frac{\varepsilon z}{2(1-v(Z))},\quad Z=z^2,
\]
which leads to P$_5$ equation
\[
 v''=\left(\frac{1}{2v}+\frac{1}{v-1}\right)(v')^2-\frac{v'}{Z}
  +\frac{(v-1)^2}{Z^2}\left(\alpha v+\frac{\beta}{v}\right)+\gamma\frac{v}{Z}+\delta\frac{v(v+1)}{v-1}
\]
with the values of parameters
\[
 \alpha= \frac{\mu^2}{2\varepsilon^2},\quad
 \beta= -\frac{(2\mu+\varepsilon-2\varepsilon\kappa)^2}{8\varepsilon^2},\quad
 \gamma= \frac{\varepsilon}{2},\quad
 \delta=0.
\]
It is known that P$_5$ with $\delta=0$ is a degenerate case which is related with $P_3$ (see \cite{GLS}), therefore it is possible to find also a substitution 
\[
 w= \frac{2p(2p-\varepsilon z)}{zp'-2\kappa p+\mu z}
\]
which brings (\ref{fp2.scalar}) to the P$^{(6)}_3$ equation (\ref{P3}) with parameters equal to
\[
 \alpha= -2\kappa-1,\quad \beta= \varepsilon(2\kappa-1)-4\mu,\quad \gamma=1, \quad \delta=-\varepsilon^2.
\]
The inverse substitution is
\[
 4p=z(w'+w^2+\varepsilon)-2\kappa w.
\]

\section*{Acknowledgements}

The authors are grateful to V.V.~Sokolov and I.A.~Bobrova for many useful discussions. 

The work was done at Ufa Institute of Mathematics with the support by the grant \#21-11-00006 of the Russian Science Foundation, https://rscf.ru/project/21-11-00006/.

\section*{Data Availability Statement}

Data sharing is not applicable to this article as no new data were created or analyzed in this study.


\end{document}